\documentclass[aps,amsmath,amssymb,twocolumn,prl,superscriptaddress]{revtex4-2}

\usepackage{graphicx,tikz}
\usepackage{dcolumn}
\usepackage{bm}
\usepackage{epsfig,hyperref,amsthm}
\usepackage{mathtools}
\usepackage{color,fancybox}
\usepackage{colordvi}
\usepackage{relsize}
\usepackage{accents}
\usepackage{wasysym} 
\usepackage{xypic}
\usepackage{enumitem}
\usepackage{mathtools,slashed}
\usepackage{times}
\usepackage{newtxmath}

\hypersetup{
	colorlinks=true,
	linkcolor=CitingColor,
	citecolor=CitingColor,
	urlcolor=CitingColor
}

\DeclareMathAlphabet\mathbfcal{OMS}{cmsy}{b}{n}
\newcommand{\ket}[1]{\ensuremath{|#1\rangle}}
\newcommand{\bra}[1]{\ensuremath{\langle #1|}}

\newcommand{\proj}[1]{\ket{#1}\bra{#1}}
\newcommand{\be}{\begin{equation}}
\newcommand{\ee}{\end{equation}}
\newcommand{\ba}{\begin{eqnarray}}
\newcommand{\ea}{\end{eqnarray}}

\newcommand{\norm}[1]{\left\|#1\right\|}

\newcommand{\id}{\mathbb{I}}

\newtheorem{theorem}{Theorem}

\newtheorem{question}{Question}

\definecolor{nred}{rgb}{0.9,0.1,0.1}
\definecolor{nblack}{rgb}{0,0,0}
\definecolor{nblue}{rgb}{0.2,0.2,0.8}
\definecolor{ngreen}{rgb}{0.2,0.6,0.2}
\definecolor{ublue}{rgb}{0,0,0.5}

\definecolor{pur}{rgb}{0.75,0,0.75}
\definecolor{nngrn}{rgb}{0,0.5,0.5}
\definecolor{CitingColor}{rgb}{0,0.3,1}

\newcommand{\blu}{\color{nblue}}

\newcommand{\CY}[1]{{\color{black}#1}}

\newcommand{\CYnewtwo}[1]{{\color{blue}#1}}

\begin{document}
\title{Dynamical Landauer Principle: Quantifying Information Transmission by Thermodynamics}

\author{Chung-Yun Hsieh}
\email{chung-yun.hsieh@bristol.ac.uk}
\affiliation{H.H. Wills Physics Laboratory, University of Bristol, Tyndall Avenue, Bristol BS8 1TL, United Kingdom}
\affiliation{ICFO - Institut de Ci\`encies Fot\`oniques, The Barcelona Institute of Science and Technology, 08860 Castelldefels, Spain}

\date{\today}

\begin{abstract}
Energy transfer and information transmission are two fundamental aspects of nature.
They are seemingly unrelated, while 
recent findings suggest that a deep connection between them is to be discovered. 
This amounts to asking:
\emph{Can we phrase the processes of transmitting classical bits equivalently as specific energy-transmitting tasks, thereby uncovering foundational links between them?}
We answer this question positively by showing that, for a broad class of classical communication tasks, a quantum dynamics' ability to transmit $n$ bits of classical information \emph{is equivalent to} its ability to transmit $n$ units of energy in a thermodynamic task.
This finding not only provides an analytical correspondence between information transmission and energy extraction tasks, but also quantifies classical communication by thermodynamics. 
Furthermore, our findings uncover the \emph{dynamical version} of Landauer's principle, showing the strong link between transmitting information and energy.
In the asymptotic regime, our results further provide thermodynamic meanings for the well-known Holevo-Schumacher-Westmoreland Theorem in quantum communication theory.
\end{abstract}

\maketitle

\section{Introduction}
Energy transfer and information transmission are two central and dynamic aspects of our everyday lives.
They underpin nature's functionalities and act as the keys to the development of sciences and technologies.
Their importance motivated scientists to build theoretical frameworks to describe them quantitatively.
From here, {\em thermodynamics}, one of our most successful theories of physics, enables us to understand energy in the forms of heat and work; {\em communication theory}, on the other hand, provides us with an analytical method to measure the abilities to transmit classical bits in different contexts.

Despite transmitting information and energy sound unrelated, there are hints that they should be closely related.
For instance, various practical communication methods use light signals to send classical bits.
By identifying light signals as certain forms of energy, a connection between energy transmission and {\em classical communication} (i.e., transmission of classical bits) is thus expected.
Clues of such a link also appeared in the theoretical ground.
Back in 1961, Landauer's seminal paper~\cite{Landauer1961} told us the thermodynamic cost of {\em erasing} one bit of information is at least $k_BT\ln2$ amount of work, revealing the link between energy and information contents, known as {\em Landauer's principle}.
Since then, scientists have investigated the implications of Landauer's principle and other thermodynamic effects on communication theory~\cite{Plenio1999,Plenio2001,Schumacher2002,Maruyama2005}. 
Nevertheless, whether one can phrase classical communication tasks equivalently as thermodynamic ones and hence bridge these two fields is still an open poblem~\cite{Plenio1999,Plenio2001,Schumacher2002,Maruyama2005,Hsieh2020,Hsieh2021}.

Armed with the recently developed tools, we eventually have the chance to solve this problem.
Since 2019, the trend of studying quantum resources possessed by {\em dynamics} provides novel tools to build the link between communication and thermodynamics.
Several recent findings from different perspectives, such as the thermodynamic effect in classical communication~\cite{Hsieh2020,Hsieh2021} and quantum memory~\cite{Narasimhachar2019} (see also Refs.~\cite{Plenio1999,Plenio2001,Schumacher2002,Maruyama2005}), classical information encoding under thermodynamic constraints~\cite{Korzekwa2019,Biswas2021}, and energy cost of information processing~\cite{Faist2015,Faist2018,Chiribella2021} jointly suggest that a fundamental and strong connection between classical information transmission and thermodynamics is to be discovered.

In this work, we provide a thermodynamic quantification, in terms of energies, of classical communication tasks (Fig.~\ref{Fig:Question}).
Foundationally, such a link can not only provide new insights to identify the roles of thermodynamics in information transmission and vice versa, but also address the dynamical counterpart of Landauer's finding.
From the practical aspects, our results can potentially shed new light on developing quantum science and technologies (see, e.g., Ref.~\cite{Auffeves2022}).

\begin{figure}
\begin{center}
\scalebox{0.8}{\includegraphics{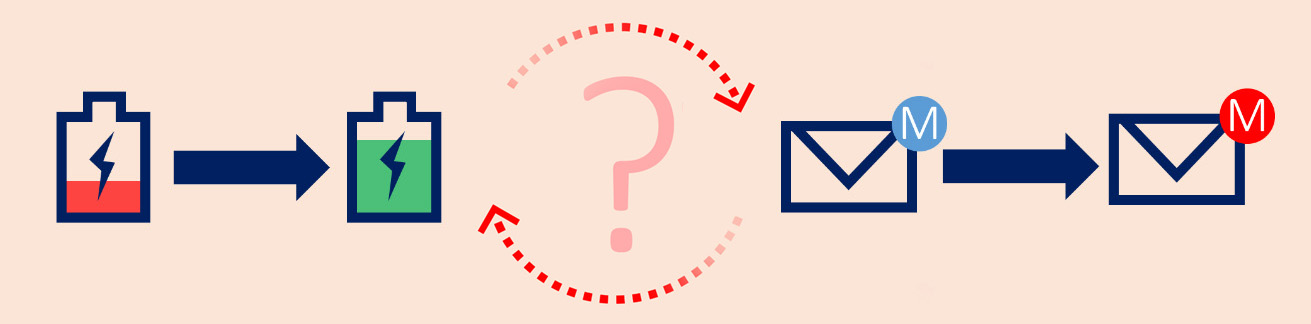}}
\caption{
{\bf Main question.}
We aim to seek the foundational relations between transmitting classical information and work-like energy.
}
\label{Fig:Question} 
\end{center}
\end{figure}

Notably, our results differ from characterising the thermodynamic cost of {\em executing quantum dynamics}.
We use two conceptual observations to illustrate the difference.
First, according to Landauer's principle for information erasure~\cite{Landauer1961,Bennett2003,delRio2011,Taranto2021,Reeb2014,Janzing2000,Sagawa2009,Maruyama2009,Anders2010,
Hilt2011,Tajima2013}, we need a high energy cost to realise the dynamics preparing a fixed pure state (i.e., a so-called state-preparation channel of some pure state) with a maximally mixed input.
However, this dynamics totally destroys the information possessed by the input state and consequently has zero ability to transmit information.
Hence, despite having a high energy cost, this energy is not consumed for information transmission.
This observation suggests that the physics behind the thermodynamic quantification of transmitting information is different from the effects of Landauer's principle.
The second observation is to examine the role of the identity map.
As reported in Ref.~\cite{Faist2015} (see also Refs.~\cite{Faist2018,Chiribella2021}), the work cost to execute a dynamics is related to the amount of logically discarded information.
Hence, the identity map is free and zero-cost to apply in this context.
On the other hand, however, the identity map gives the best communication performance and is the most valuable resource in any nontrivial communication task.
Hence, the physics behind the thermodynamic quantification of information transmission conceptually differs from the energy cost of a quantum process~\CYnewtwo{\cite{Faist2015,Faist2018,Chiribella2021,Chiribella2022}}.

Since we aim to bridge the transmissions of information and work-like energy (Fig.~\ref{Fig:Question}), we start with formalising them.

\section{Results}

\subsection{Classical communication tasks of quantum channels}
Classical information can be described by a set of integers $\{m\}_{m=0}^{M-1}$.
Intuitively, if one can use physical processes to send this data set to a remote agent reliably, it is then possible for them to {\em communicate with} each other.
Hence, the goal is to formally describe how to transmit the classical data through a quantum dynamics, denoted by $\mathcal{N}$.
Mathematically, $\mathcal{N}$ is called a {\em channel}, the so-called completely-positive trace-preserving linear maps~\cite{QIC-book}.
To this end, we need to first {\em encode} the classical information into a set of quantum states $\{\rho_{m}\}_{m=0}^{M-1}$.
The channel can then act on them and proceed with the transformation.
After that, the receiver needs to {\em decode} the classical information from the output quantum system.
This can be achieved by applying a measurement, described by a {\em positive operator-valued measure} (POVM)~\cite{QIC-book} $\{E_m\}_{m=0}^{M-1}$ with $\sum_{m=0}^{M-1}E_{m}=\id$ and $E_m\ge0$ in the output space.
If the receiver's measurement outcome coincides with the originally encoded classical index, the transmission is successful.
In general, apart from the encoding and decoding, we can further consider an additional assistance structure.
This can be formally written as a set of {\em superchannels}~\cite{Chiribella2008,Chiribella2008-2}, denoted by $\Theta$, where a superchannel is a linear map bringing a channel to another channel.
We allow the sender and receiver to use some superchannel $\Pi$ from $\Theta$ to assist the communication.
This amounts to replacing the transmission channel $\mathcal{N}$ with $\Pi(\mathcal{N})$.
For every $m$, the transformation looks like
$
m\mapsto{\rm tr}\left[E_m\Pi(\mathcal{N})(\rho_m)\right].
$
We call this a {\em $\Theta$-assisted scenario}.
In the literature, a standard measure of $\mathcal{N}$'s performance in such a scenario is the highest amount of classical data that can be correctly transmitted up to some acceptable error.
This is called {\em one-shot $\Theta$-assisted classical capacity with error $\epsilon$}, which is given by [see Figure~\ref{Fig:SchematicInterpretation} (a) for a schematic illustration]
\begin{align}\label{Eq:Theta capacity}
C_{\Theta,(1)}^\epsilon(\mathcal{N})\coloneqq\max_{\Pi\in\Theta,\{\rho_m\},\{E_m\}}\log_2M,
\end{align}
where the maximisation is taken over every $\Theta$-assisted scenario with the average success probability higher than $1-\epsilon$, i.e., 
$
\sum_{m=0}^{M-1}\frac{1}{M}{\rm tr}\left[E_m\Pi(\mathcal{N})(\rho_m)\right]\ge1-\epsilon.
$
(see Appendix I for details; in the companion paper~\cite{CompanionPRA2}, we provide a complete mathematical treatment for pedagogical purposes).
$C_{\Theta,(1)}^\epsilon(\mathcal{N})$ tells us the largest size of communicable classical data via $\mathcal{N}$ with assistance from $\Theta$, up to probability at most $\epsilon$ to wrongly decode the classical information.
It is called {\em one-shot} because the channel $\mathcal{N}$ only serves {\em once}, contrary to the {\em asymptotic limit} involving infinitely many copies of $\mathcal{N}$.

Several existing scenarios are included as special cases.
First, when $\Theta = \Theta_{\rm C} \coloneqq \{\Pi(\mathcal{N}) = \mathcal{N}\;\forall\mathcal{N}\}$ containing only identity superchannels, it gives the (standard) one-shot classical capacity $C_{\Theta_{\rm C},(1)}^\epsilon = C_{(1)}^\epsilon$~\cite{Wang2013,Renes2011}. 
It characterises the channel's ``primal'' ability to transmit classical information.
When $\Theta = \Theta_{\rm LOSE}$ is the set of {\em local operation plus pre-shared entanglement} (LOSE) superchannels~\cite{Beckman2001}, it gives the one-shot {\em entanglement-assisted classical communication}~\cite{Datta2013,Anshu2019,Matthews2014}.
Finally, when $\Theta=\Theta_{\rm NS}$ containing all no-signaling superchannels, it is the one-shot {\em non-signalling assisted classical communication}~\cite{Beckman2001,Eggeling2002,Piani2006,Duan2016,Takagi2020}.

\begin{figure}
\begin{center}
\scalebox{0.8}{\includegraphics{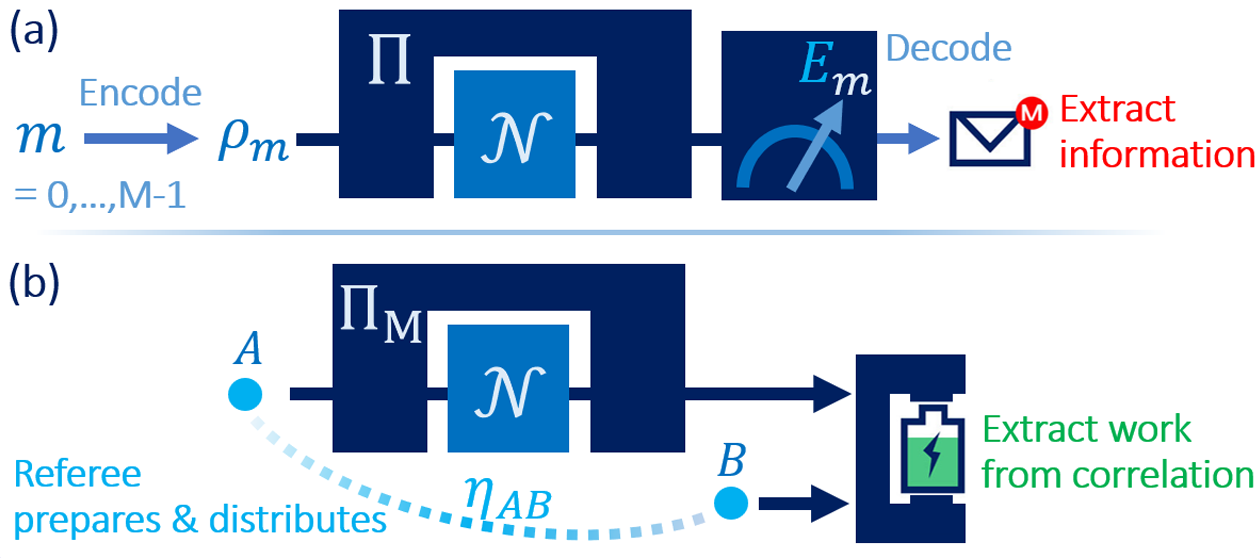}}
\caption{
{\bf Operational tasks.}
(a) The one-shot $\Theta$-assisted classical capacity describes a communication task where the sender encodes classical data $\{m\}_{m=0}^{M-1}$ into states $\{\rho_m\}_{m=0}^{M-1}$.
After transmission via $\mathcal{N}$ assisted by some $\Pi\in\Theta$, the receiver decodes the transmitted classical data by a measurement $\{E_m\}_{m=0}^{M-1}$.
(b) The one-shot $\Theta$-assisted $\epsilon$-deterministic genuinely transmitted energy measures the energy that can {\em only} be transmitted by $\mathcal{N}$.
The referee prepares $\eta_{AB}$ and distributes it to the sender and receiver. 
The sender sends their part $A$ to the receiver via $\Pi_M(\mathcal{N})$, a classical channel induced by $\mathcal{N}$ and $\Pi\in\Theta$. 
Then, the receiver extracts work from the bipartite correlation. 
This can be done by changing Hamiltonians without changing states (i.e., {\em quenches}) to make the bipartite output locally thermal and extracting work from that bipartite state (detailed in Appendix II).}
\label{Fig:SchematicInterpretation} 
\end{center}
\end{figure}

\subsection{Energy transmission tasks of quantum channels}
To quantify information transmission thermodynamically, we now use an operational task to analyse the (work-like) energy that is {\em definitely transmitted} by a channel $\mathcal{N}$.
Consider a setting with referee, sender, and receiver [Figure~\ref{Fig:SchematicInterpretation} (b)]. 
At the beginning, the referee prepares a bipartite state $\eta_{AB}$ diagonal in a given computational basis \mbox{$\{\ket{m}_A\otimes\ket{n}_{B}\}_{n,m=0}^{M-1}$} with some initial local Hamiltonians (both $A,B$ are $M$-dimensional).
Apart from being of finite-energy, there is no other restriction on these Hamiltonians.
The referee distributes $\eta_{AB}$ to the sender (with ``part $A$'') and receiver (with ``part $B$'').
Then, the sender uses a {\em classical} channel in $A$ (i.e., its outputs are always diagonal in $\{\ket{m}_A\}_{m=0}^{M-1}$) induced by $\mathcal{N}$ and $\Theta$ to locally transmit $\eta_{AB}$'s part $A$ to the receiver [see also Eq.~\eqref{Eq:ClassicalVersion} in Appendix I]. 
After that, the receiver possesses a bipartite state. 
Now, we aim to identify the amount of energy that is {\em transmitted}, rather than {\em created} by the channel.
For instance, when $\mathcal{N}$ is an erasure channel $(\cdot)\mapsto\proj{0}$, it can easily {\em create} extractable energy that is clearly not transmitted at all.
Hence, generally, the extractable work from part $A$ does not {\em solely} result from transmission.
The extractable work from part $B$ is also ruled out—it was with the receiver {\em before} even applying the channel.
To isolate the {\em genuinely transmitted} energy, the receiver thus needs to check the extractable work from the output's {\em bipartite correlation}---as this energetic contribution {\em is not} obtainable by the receiver {\em before} the channel, and it {\em cannot} be generated by channels locally acting on $A$ (see companion paper~\cite{CompanionPRA2} for a detailed proof).
To measure this, we adopt \AA berg's approach~\cite{Aberg2013} to write the (one-shot $\epsilon$-deterministic) extractable work from a state $\eta_{AB}$'s correlation as $W_{\rm corr,(1)}^\epsilon(\eta_{AB})$, which isolates the genuinely global contribution to work-like energy, as local ones have been deduced [see Appendix II and Eq.~\eqref{Eq:WcorrState} for details].
Then, we measure the work-like energy that can {\em only} result from transmission by the {\em one-shot $\Theta$-assisted \mbox{$\epsilon$-deterministic} genuinely transmitted energy} of $\mathcal{N}$, which is defined by
\begin{align}\label{Eq:ThermoTask}
W^{\epsilon}_{\rm corr|\Theta,(1)}(\mathcal{N})\coloneqq\sup_{\substack{M\in\mathbb{N},\eta_{AB}\\\Pi_M\in\Theta_M}}W_{\rm corr,(1)}^\epsilon\left[(\Pi_M(\mathcal{N})_A\otimes\mathcal{I}_{B})(\eta_{AB})\right],
\end{align} 
where 
$
\Theta_M\coloneqq\left\{\mathcal{E}_{{\rm de}|M}\circ\Pi(\cdot)\circ\mathcal{E}_{{\rm en}|M}\,\middle|\,\Pi\in\Theta,\mathcal{E}_{{\rm de}|M},\mathcal{E}_{{\rm en}|M}\right\}.
$
Here, $\mathcal{E}_{{\rm en}|M}:A\to S_{\rm in}$ ($\mathcal{E}_{{\rm de}|M}:S_{\rm out}\to A$) is a {\em classical-to-quantum} ({\em quantum-to-classical}) channel.
$S_{\rm in}$ ($S_{\rm out}$) is $\Pi(\mathcal{N})$'s input (output) space.
$\Pi_M(\mathcal{N})$'s are classical channels in $A$ induced by $\mathcal{N}$ and $\Theta$.
We also call them {\em classical versions of $\mathcal{N}$ (assisted by $\Theta$)}.
See Appendix I for details.

\subsection{Quantifying communication by energy transmission}
As our first main result, we report the following bounds, whose complete version is given in Appendix III (Theorem~\ref{Result:MainResult}):
{\em For a channel $\mathcal{N}$, a background temperature \mbox{$0<T<\infty$}, and errors $0<\delta\le\omega<\epsilon\le1-1/\sqrt{2}$, we have that}
\begin{align}\label{Eq:MainResultInformalForm}
 W_{\rm corr|\Theta,(1)}^\omega(\mathcal{N}) \lesssim (k_BT\ln2)C_{\Theta,(1)}^\epsilon(\mathcal{N})\le W_{\rm corr|\Theta,(1)}^{\epsilon+\delta}(\mathcal{N}).
\end{align}
The inequality ``$\lesssim$'' is up to one-shot error terms \mbox{$\sim O\left(\log_2(1/\epsilon)\right)$}---contributions from such terms vanish when we consider sufficiently many (but finite) copies of the channel and average the figure-of-merits over the number of copies.
In this sense, two one-shot figure-of-merits can carry the same physical meaning if their difference is a one-shot error term.
Hence, we conclude that (see also Fig.~\ref{Fig:Result})
\begin{center}
{\em 
$W_{\rm corr|\Theta,(1)} \approx (k_BT\ln2)C_{\Theta,(1)}$ up to one-shot error terms.
}
\end{center}
This means a broad class of classical communication scenarios, as characterised by $C_{\Theta,(1)}$, are {\em physically equivalent to} energy transmission tasks described by $W_{\rm corr|\Theta,(1)}$.
This gives the first quantitative correspondence between one-shot communication and small-scale thermodynamics.

\begin{figure}
\begin{center}
\scalebox{0.8}{\includegraphics{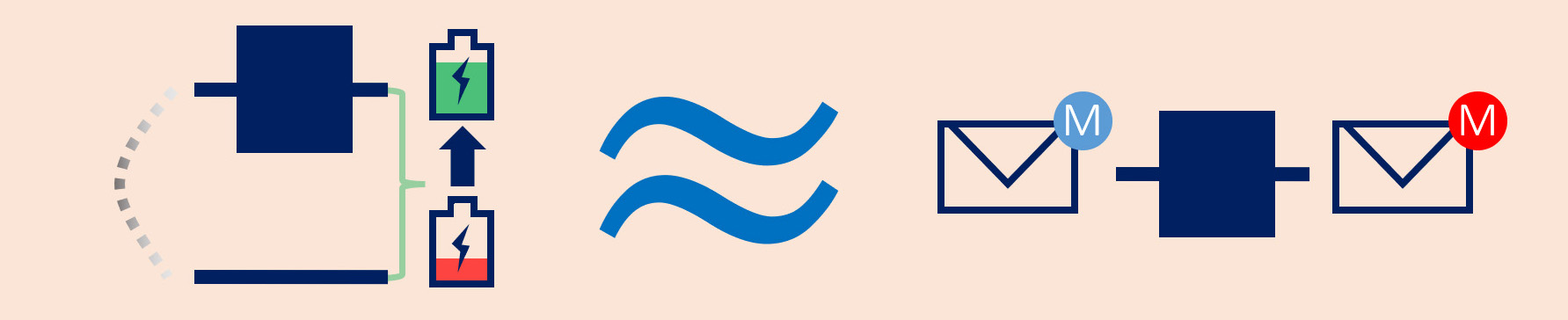}}
\caption{
{\bf Main result.}
The abilities to transmit information and energy are equivalent for quantum dynamics.}
\label{Fig:Result} 
\end{center}
\end{figure}

\subsection{Dynamical Landauer's principle}
Importantly, Eq.~\eqref{Eq:MainResultInformalForm} uncovers a {\em dynamical version} of Landauer's principle~\cite{Landauer1961}. 
Loosely speaking, Landauer's principle states that preparing a qubit pure state from the maximally mixed one must be accompanied by at least $k_BT\ln2$ energy cost.   
Together with Szilard engine~\cite{Szilard1929}, they equate informational and energetic properties of states---a state carries one bit of deterministic information (e.g., a qubit pure state) {\em if and only if} it possesses one unit of extractable work (i.e., $k_BT\ln2$).
Here, we equate informational and energetic properties of dynamics, which are the abilities to transmit information and energy.
To see this, consider again the $\Theta$-assisted scenario (Fig.~\ref{Fig:SchematicInterpretation}).
For a given dynamics (described by a channel $\mathcal{N}$), when we view it information-theoretically as a communication channel, it can transmit $n$ bits of information if and only if $n\le C_{\Theta,(1)}^\epsilon(\mathcal{N})$.
In the same physical setting, we can also view it thermodynamically as an energy-transmitting process [as in Eq.~\eqref{Eq:ThermoTask}]. Then, Eq.~\eqref{Eq:MainResultInformalForm} implies that, necessarily and sufficiently, the same dynamics must be able to transmit $n\times(k_BT\ln2)$ energy (up to one-shot error terms).
Hence, we conclude that:
\begin{center}
{\em The ability to transmit $n$ bits of information is equivalent to the ability to transmit $n\times(k_BT\ln2)$ energy.}
\end{center}

Crucially, the equivalence between the {\em abilities} to do two things may not imply the equivalence between these two things. 
We now argue that transmitting information and energy can happen {\em simultaneously} and can be two facets of the {\em same} physical process.
Here, we provide a back-of-the-envelope argument, and the complete mathematical proof is allocated in the companion paper~\cite{CompanionPRA2}. 
Consider a channel $\mathcal{N}$ in the energy transmission task [Fig.~\ref{Fig:SchematicInterpretation} (b)] where the referee prepares 
\mbox{$\eta_{AB}=\Phi_{AB} \coloneqq \frac{1}{M}\sum_{m=0}^{M-1}\proj{m}_{A}\otimes\proj{m}_{B}$}.
Physically, this state is prepared as a statistical mixture of \mbox{$\ket{m}_A\otimes\ket{m}_B$} in a multi-trials experiment---during each trial, with (uniform) probability $1/M$, the referee prepares and sends $\ket{m}_A$ ($\ket{m}_B$) to the sender (receiver).
We now argue that, in this setting, transmitting $n = \log_2M$ bits of information {\em must} be accompanied by $n\times(k_BT\ln2)$ transmitted energy.
First, transmitting $n$ bits of information means there is some $\Pi_M$ such that $\Pi_M(\mathcal{N})(\proj{m})\approx_\epsilon\proj{m}$ \mbox{for every $m$} [``$\approx_\epsilon$'' means it is up to terms $\sim O(\epsilon)$].
Hence, for most trials, the receiver's bipartite outputs $\approx_\epsilon\ket{m}\otimes\ket{m}$.
After sufficiently many trials, the receiver's bipartite output can be described by a statistical mixture $\approx_\epsilon\Phi_{AB}$.
This is a {\em necessary} condition if $\mathcal{N}$ transmits $n$ bits of information in the current setting.
To see if energy is transmitted, one can implement work extraction protocols {\em after} the information transmission.
The extractable work from $\Phi_{AB}$'s correlation is $n\times(k_BT\ln2)$ [see Eq.~\eqref{Eq:WcorrState}]. 
This is, up to terms $\sim O(\epsilon)$, the energy {\em definitely transmitted} by $\mathcal{N}$, as we argue above Eq.~\eqref{Eq:ThermoTask}. 
Hence, $\mathcal{N}$ is transmitting information {\em and} energy {\em at the same time}, and the amount of transmitted energy is {\em lower bounded} by the amount of transmitted information.
From here, we obtain a truly work-like, genuinely dynamical version of Landauer's principle:
\begin{center}
{\em
In the above setting, transmitting $n$ bits of information must be accompanied by transmitting $n\times(k_BT\ln2)$ energy.} 
\end{center}
Notably, the transmitted energy is extracted by additional processes {\em after} information transmission. 
Hence, interestingly:
\begin{center}
{\em We can convert information transmission into energy transmission.}
\end{center}
This can be viewed as a dynamical counterpart of the famous information-to-work conversion via Szilard engine~\cite{Szilard1929}.
In other words, the energy transmission is {\em mediated by} information.
We thus fully answer this work's central question.

\subsection{Implications of thermodynamics for communication, and implications of communication for thermodynamics}
We discuss some implications of Eq.~\eqref{Eq:MainResultInformalForm}.
When $\Theta = \Theta_{\rm C}$, both upper and lower bounds in Eq.~\eqref{Eq:MainResultInformalForm} are potentially tighter than the ones reported in Ref.~\cite{Wang2013} (see also Theorem~2 in companion paper~\cite{CompanionPRA2}).
When $\Theta = \Theta_{\rm LOSE}$, the corresponding one-shot classical capacity is called {\em entanglement-assisted}, denoted by $C_{\rm EA,(1)}^\epsilon$~\cite{Datta2013,Anshu2019,Matthews2014}.
Equation~\eqref{Eq:MainResultInformalForm} implies that, up to one-shot error terms, $C_{\rm EA,(1)}^\epsilon(\mathcal{N})\approx W_{\rm corr|\Theta_{\rm LOSE},(1)}^\epsilon(\mathcal{N})$, which is the transmitted energy assisted by LOSE superchannels.
Using this, Eq.~\eqref{Eq:MainResultInformalForm} thus gives $C_{\rm EA,(1)}^\epsilon$'s entropic bounds reported in Refs.~\cite{Datta2013,Anshu2019,Matthews2014} thermodynamic meanings in 
energy transmission.
The same argument also works when $\Theta = \Theta_{\rm NS}$, and the one-shot {\em non-signalling-assisted} classical capacity $C_{\rm NS,(1)}^\epsilon$ is related to $W_{\rm corr|\Theta_{\rm NS},(1)}^\epsilon(\mathcal{N})$, offering a link between thermodynamics and the entropic bounds found in Refs.~\cite{Beckman2001,Eggeling2002,Piani2006,Duan2016,Takagi2020}.
Finally, once we know how to estimate $C_{\Theta,(1)}^\epsilon$, Eq.~\eqref{Eq:MainResultInformalForm} provides a way to estimate $W^{\epsilon}_{\rm corr|\Theta,(1)}$; namely, communication-theoretic results can 
apply to thermodynamic questions.

\subsection{Asymptotic regime and thermodynamic meaning of Holevo-Schumacher-Westmoreland theorem}
After fully addressing the one-shot cases, let us now go to the asymptotic regime.
Considering multi-copy and taking the average, we can recover the {\em asymptotic} classical capacity assisted by different structures $\Theta$.
Thus, Eq.~\eqref{Eq:MainResultInformalForm} can also provide thermodynamic bounds in the asymptotic limit.
First, we define {\em $\Theta$-assisted classical capacity} as
\mbox{$
C_\Theta(\mathcal{N})\coloneqq\lim_{\epsilon\to0}\varliminf_{k\to\infty}\frac{1}{k}C_{\Theta,(1)}^\epsilon\left(\mathcal{N}^{\otimes k}\right).
$}
This notion recovers various scenarios such as the standard, entanglement-assisted, and non-signalling assisted scenarios by selecting different $\Theta$.
Now, one can apply the independent identically distributed (iid) limit and obtain the asymptotic version of Eq.~\eqref{Eq:MainResultInformalForm}:
{\em For a channel $\mathcal{N}$ and a background temperature $0<T<\infty$, we have that}
\begin{align}\label{Eq:MainResultAsympVer}
(k_BT\ln2)C_\Theta(\mathcal{N}) = W_{\rm corr|\Theta}(\mathcal{N}),
\end{align}
where
$W_{\rm corr|\Theta}(\mathcal{N})\coloneqq\lim_{\epsilon\to0}\varliminf_{k\to\infty}\frac{1}{k}W_{\rm corr|\Theta,(1)}^\epsilon\left(\mathcal{N}^{\otimes k}\right)$
can be understood as the average transmitted energy in the asymptotic limit.
Hence, the physical messages we obtained in the one-shot regime can directly apply to the asymptotic limit as an {\em equality} without any error term.
To illustrate \CY{Eq.~\eqref{Eq:MainResultAsympVer}}, we use the standard classical capacity as an example, and the same argument works for other choices of $\Theta$.
Setting $C_{\Theta_{\rm C}} = C$, the classical capacity is then given by $C(\mathcal{N})\coloneqq\lim_{\epsilon\to0}\varliminf_{k\to\infty}\frac{1}{k}C_{(1)}^\epsilon\left(\mathcal{N}^{\otimes k}\right)$~\cite{Wang2013}.
Then Eq.~\eqref{Eq:MainResultAsympVer} implies that 
\begin{align}\label{Eq:ThermoHSW}
(k_BT\ln2)C(\mathcal{N})= W_{\rm corr|\Theta_{\rm C}}(\mathcal{N}).
\end{align}
Equation~\eqref{Eq:ThermoHSW} can be interpreted as a thermodynamic version of the famous {\em Holevo-Schumacher-Westmoreland} (HSW) Theorem~\cite{Holevo1973,Holevo1998,Schumacher1997,Wilde-book}, which equates classical capacity and the regularised version of Holevo information.
Here, as a consequence of Eq.~\eqref{Eq:MainResultAsympVer}, the asymptotic transmitted energy $W_{\rm corr|\Theta_{\rm C}}$ {\em equals} the classical capacity in the unit of $k_BT\ln2$.

\subsection{No-go results for classical communication with non-equilibrium constraints}
To illustrate another physical implication of Eq.~\eqref{Eq:MainResultInformalForm}, let us consider an alternative figure-of-merit $C_{(1)|\text{$\theta$-equi}}^{\epsilon}$ with a parameter \mbox{$0\le\theta\le1/2$}, which is $C_{(1)}^\epsilon$ subject to an additional thermodynamic constraint \mbox{$\norm{\mathcal{E}_{{\rm de}|M}\circ\mathcal{N}\circ\mathcal{E}_{{\rm en}|M}(\id_M/M) - \id_M/M}_1\le2\theta$}. 
Namely, the encoding and decoding cannot drive the system out of thermal equilibrium more than $\theta$ under trace norm.
Here, $\id_M/M$ describes thermal equilibrium when the system Hamiltonian is fully degenerate, and any other state possesses the so-called {\em informational non-equilibrium}~\cite{Gour2015,Stratton2023,Hsieh2024-3}. 
Hence, the parameter $\theta$ controls the informational non-equilibrium $\mathcal{N}$ can generate via encoding and decoding.
Defining the asymptotic form as \mbox{$C_{\text{$\theta$-equi}}(\mathcal{N})\coloneqq\lim_{\epsilon\to0}\varliminf_{k\to\infty}\frac{1}{k}C_{(1)|\text{$\theta$-equi}}^\epsilon\left(\mathcal{N}^{\otimes k}\right)$.}
We further consider two different asymptotic figure-of-merits, which are \mbox{$C_{\rm max}(\mathcal{N})\coloneqq\sup_{0<\theta<1/2}C_{\text{$\theta$-equi}}(\mathcal{N})$} and \mbox{$C_{\rm min}(\mathcal{N})\coloneqq\lim_{\theta\to0}C_{\text{$\theta$-equi}}(\mathcal{N})$.}
Note that they have very different physical meanings.
$C_{\rm max}$ is the highest amount of transmissible classical information among all possible encoding/decoding's abilities to generate informational non-equilibrium.
On the other hand, $C_{\rm min}$ is the somewhat limited classical capacity requesting encoding/decoding to asymptotically have no ability to generate informational non-equilibrium.
Hence, by construction, we have \mbox{$C_{\rm min}\le C_{\rm max}$}.
When the opposite inequality also holds, we say the {\em strong converse property}~\cite{Takagi2020} occurs in this setting.
Interestingly, for $2\epsilon\le1-1/\sqrt{2}$, we can prove that \mbox{$C_{(1)}^\epsilon(\mathcal{N})\le C_{(1)|\text{$\theta$-equi}}^{\epsilon}(\mathcal{N})\le W_{\rm corr|\Theta_{\rm C},(1)}^{2\epsilon}(\mathcal{N})/(k_BT\ln2)$} {\em for every $\epsilon\le\theta<1/2$} (see Proposition~2 in the companion paper~\cite{CompanionPRA2}).
Using Eq.~\eqref{Eq:MainResultInformalForm}, we obtain $C_{\text{$\theta$-equi}}(\mathcal{N}) = C(\mathcal{N})$ {\em for every $0<\theta<1/2$}. 
Consequently, for every channel $\mathcal{N}$,
\begin{align}
C_{\rm min}(\mathcal{N}) = C_{\rm max}(\mathcal{N}) = C(\mathcal{N}).
\end{align}
Hence, Eq.~\eqref{Eq:MainResultInformalForm} provides a general {\em no-go result}---out-of-equilibrium effect from encoding/decoding {\em cannot enhance} the ability to transmit classical information asymptotically in the present setting.
This also means that optimal information transmission does not require the ability to generate informational non-equilibrium---the {\em preservability} of informational non-equilibrium~\cite{Stratton2023,Hsieh2020,Hsieh2021} is the key resource for information transmission.
Finally, as another implication, one can use Eq.~\eqref{Eq:MainResultInformalForm} to prove HSW Theorem {\em with Holevo information subject to an additional thermodynamic constraint}, and similar strong converse property as well as no-go interpretation can be obtained.
See the companion paper~\cite{CompanionPRA2} for details.

\section{Discussions}
We remark that an early work by Plenio~\cite{Plenio1999} has mentioned the connection between Holevo bound and Landauer erasure process (see also the early observations of the connection between communication and thermodynamics in Refs.~\cite{Plenio2001,Schumacher2002,Maruyama2005}).
Here, we aim to bridge transmitting information and energy, which is different from the previous works.

Many open questions remain, and here we list a few.
First, recent results have addressed the carriable amount of classical information of a {\em state} under thermodynamic constraints~\cite{Narasimhachar2019,Korzekwa2019,Biswas2021}, and it is rewarding to explore the relation between these results and our findings.
Second, it would be interesting to further explore the link between thermodynamics and communication in the dynamical resource theory of informational non-equilibrium~\cite{Stratton2023,Hsieh2020,Hsieh2024-3} and device-independent frameworks~\cite{Chen2024PRL,Hsieh2024,Hsieh2024-2}.
Third, inspired by a recent work studying thermal operation subject to non-ideal baths with fluctuation~\cite{Shu2019}, it would be useful to discuss how robust our results can be when noises and fluctuation occur.
Finally, it is unknown whether the idea of stochastic distillation via post-selection~\cite{Hsieh2023,Ku2023,Ku2022,Hsieh2024,Hsieh2024-2,Regula2022Quantum,Regula2022PRL,Takagi2024PRA,Yuan2024PRL,Ji2024} and exclusion-type tasks~\cite{Ducuara2022,Ducuara2023,Ducuara2023-2,Hsieh2023-2,Stratton2024} can be applied to the relation of thermodynamics and communication and hence discover novel thermodynamic tasks.

\section{Acknowledgements}
We thank Antonio Ac\'in, Alvaro Alhambra, Stefan B$\ddot{\rm a}$uml, Philippe Faist, Yeong-Cherng Liang, Matteo Lostaglio, Jef Pauwels, Mart\'i Perarnau-Llobet, Bartosz Regula, Valerio Scarani, Gabriel Senno, Yaw-Shih Shieh, Paul Skrzypczyk, Jacopo Surace, Gelo Noel M. Tabia, Ryuji Takagi, Philip Taranto, and Armin Tavakoli for fruitful discussions.
We also thank the Quantum Thermodynamics Summer School (23-27 August 2021, Les Diablerets, Switzerland), organised by L\'idia del Rio and Nuriya Nurgalieva, for the inspirational environment that helped me to improve the early version of this work significantly.
We acknowledge support from ICFOstepstone (the Marie Sk\l odowska-Curie Co-fund GA665884), the Spanish MINECO (Severo Ochoa SEV-2015-0522), the Government of Spain (FIS2020-TRANQI and Severo Ochoa CEX2019-000910-S), Fundaci\'o Cellex, Fundaci\'o Mir-Puig, Generalitat de Catalunya (SGR1381 and CERCA Programme), the ERC Advanced Grant (on grants CERQUTE and FLQuant), the AXA Chair in Quantum Information Science, the Royal Society through Enhanced Research Expenses (on grant NFQI), and the Leverhulme Trust Early Career Fellowship (on grant ``Quantum complementarity: a novel resource for quantum science and technologies'' with number ECF-2024-310).

\section{Appendix}

\subsection{Appendix I: Classical communication via channels}
First, for every $m$, the mapping 
\mbox{$
m\mapsto{\rm tr}\left[E_m\Pi(\mathcal{N})(\rho_m)\right]
$}
can be equivalently described as (see the companion paper~\cite{CompanionPRA2} Section~II.A for details)
\begin{align}\label{Eq:InputOutput}
m\mapsto\bra{m}\left[\mathcal{E}_{{\rm de}|M}\circ\Pi(\mathcal{N}_S)\circ\mathcal{E}_{{\rm en}|M}\right](\proj{m})\ket{m},
\end{align}
where $\{\ket{m}\}_{m=0}^{M-1}$ is an orthogonal {\em basis} of an $M$-dimensional system (denoted by $A$), $\mathcal{E}_{{\rm en}|M}:{A}\to S_{\rm in}$ is a {\em classical-to-quantum} channel of the form $\sum_{m=0}^{M-1}\rho_m\bra{m}(\cdot)_{A}\ket{m}$, and $\mathcal{E}_{{\rm de}|M}:S_{\rm out}\to{A}$ is a {\em quantum-to-classical} channel of the form $\sum_{m=0}^{M-1}\proj{m}_{A}{\rm tr}\left[E_m(\cdot)_{S_{\rm out}}\right]$.
Here, $S_{\rm in}$ and $S_{\rm out}$ represent the input and output spaces of $\Pi(\mathcal{N})$, respectively.
Also, the term ``classical'' means that only diagonal contribution (in the system {$A$}) matters.
Hence, to understand $\mathcal{N}$'s ability to send classical data with size $M$, it suffices to consider a ``classical version'' of the channel $\Pi(\mathcal{N})$ in a $M$-dimensional system with a fixed basis $\{\ket{m}\}_{m=0}^{M-1}$.
More precisely, we define
\begin{align}\label{Eq:ClassicalVersion}
\Theta_M\coloneqq\left\{\mathcal{E}_{{\rm de}|M}\circ\Pi(\cdot)\circ\mathcal{E}_{{\rm en}|M}\,\middle|\,\Pi\in\Theta,\mathcal{E}_{{\rm de}|M},\mathcal{E}_{{\rm en}|M}\right\}.
\end{align}
This set contains all realisations of $\Theta$'s members in $A$ such that we only care about diagonal terms in {\em both} input and output spaces.
For the given $\mathcal{N}$, the set $\{\Pi_M(\mathcal{N})\,|\,\Pi_M\in\Theta_M\}$ includes all realisations of $\mathcal{N}$ assisted by $\Theta$ in such a setting.
Hence, we call them {\em classical versions of $\mathcal{N}$ assisted by $\Theta$}, which are the relevant notions in information transmission.

The general input-output relation given by Eq.~\eqref{Eq:InputOutput} characterises a classical communication scenario assisted by the allowed structure $\Theta$, which is the {\em $\Theta$-assisted scenario} mentioned in the main text.
Hence, the one-shot $\Theta$-assisted classical capacity with error $\epsilon$ defined in Eq.~\eqref{Eq:Theta capacity}, namely, $C_{\Theta,(1)}^\epsilon(\mathcal{N})$, can be rewritten as the following maximisation:
\begin{align}\label{Eq:DefClassicalCapacity}
\max\{\log_2M\,|\,\Pi_M\in\Theta_M,P_s[\Pi_M(\mathcal{N})]\ge1-\epsilon\},
\end{align}
where
$
P_s[\Pi_M(\mathcal{N})]\coloneqq\sum_{m=0}^{M-1}\frac{1}{M}\bra{m}\Pi_M(\mathcal{N})(\proj{m})\ket{m}
$
is the averaged success probability.
See companion paper~\cite{CompanionPRA2} Section~II.A for a complete mathematical framework.

Note that certain choices of $\Theta$ can lead to trivial cases. For instance, if the superchannels can send an additional signal from the sender to the receiver, it trivialises the communication setup since one can always transfer the classical data through the member of $\Theta$ without touching the main channel. This also explains why physically $\Theta$ is usually assumed to be a subset of all non-signalling superchannels.

\subsection{Appendix II: Extractable work from correlation}
Now, we will briefly recall the work extraction scenario from \AA berg's approach~\cite{Aberg2013} (see also the companion paper~\cite{CompanionPRA2} Section~IV.A).
First, for a system with Hamiltonian $H$ in contact with a large heat bath in temperature $T$, the thermal equilibrium is described by the {\em thermal state} 
$
\gamma_H~\coloneqq~e^{-\frac{H}{k_BT}}/{\rm tr}\left(e^{-\frac{H}{k_BT}}\right),
$
where $k_B$ is the Boltzmann constant.
To address work extraction, suppose the system is in an energy-incoherent state $\rho$; namely, it satisfies $\left[\rho,\gamma_H\right]=0$. 
Then one can stay in the energy eigenbasis and treat both $\rho$ and $\gamma_H$ classically.
In particular, they can be viewed as random variables which output eigenenergies according to their probability distributions in the energy eigenbasis.
The goal is to apply certain physical processes to transform $\rho$ into some final state, usually thermal, and extract reusable energy, also known as {\em work}.
Intuitively, not every transformation is allowed.
In \AA berg's scenario~\cite{Aberg2013}, allowed processes are combinations of the following two kinds of operations.
The first is tuning the energy levels of the system's Hamiltonian without changing the state (isentropic process; also known as {\em quench}). 
The second is thermalisation, which will bring the input state to the thermal state with the corresponding Hamiltonian and a given background temperature $0<T<\infty$ (i.e., putting the system in contact with a large enough heat bath in this temperature so that the system reaches thermal equilibrium).
The energy gain during operations of the first kind defines extractable work, which is again treated as a random variable. 
Note that thermalisation is assumed to contribute no work gain.
Now, following Ref.~\cite{Aberg2013}, a random variable $X$ is said to have a {\em $(\epsilon,\delta)$-deterministic value $x$} if the probability to find $|X-x|\le\delta$ is strictly larger than $1-\epsilon$.
In other words, $\delta$ represents the precision that the random variable can take the value $x$, and $1-\epsilon$ is the probability of doing so.
Let $W_{\rm ext,(1)}^{(\epsilon,\delta)}\left(\rho,H\right)$ denote the highest possible $(\epsilon,\delta)$-deterministic extractable work among all allowed processes acting on initial state $\rho$ with the same initial and final Hamiltonian $H$.
Then the quantity $W_{\rm ext,(1)}^{\epsilon}\left(\rho,H\right)\coloneqq\lim_{\delta\to0}W_{\rm ext,(1)}^{(\epsilon,\delta)}\left(\rho,H\right)$, which is called the {\em one-shot $\epsilon$-deterministic extractable work of $\rho$ subject to $H$}, characterises the highest extractable work with arbitrarily well precision, up to a failure probability no greater than $\epsilon$.
It has been shown that, for $0<\epsilon\le1-1/\sqrt{2}$~\cite{Aberg2013},
\begin{align}\label{Eq:Aberg}
0\le \frac{W_{\rm ext,(1)}^\epsilon\left(\rho,H\right)}{k_BT\ln2} - D_0^\epsilon\left(\rho\,\|\,\gamma_H\right)\le \log_2\frac{1}{1-\epsilon},
\end{align}
where $D_0^\epsilon$ is the {\em $\epsilon$-smoothed relative R\'enyi $0$-entropy}~\cite{Aberg2013} that is defined for two commuting states $\rho=\sum_jq_j\proj{j},\eta=\sum_jr_j\proj{j}$ as
$
D_0^\epsilon\left(\rho\,\|\,\eta\right)\coloneqq\max_{\Lambda:\sum_{j\in\Lambda}q_j>1-\epsilon}\log_2\frac{1}{\sum_{j\in\Lambda}r_j}.
$

Now we adopt \AA berg's scenario to discuss work extraction from {\em classical correlation}.
With a given bipartite state $\rho_{AB}$ correlating $A$ and $B$, we want to know the extractable work {\em from its correlation only.}
Following the argument of Ref.~\cite{Perarnau-Llobet2015}, this can be done if we design the local Hamiltonians so that this state becomes locally thermal; i.e., its reduced states in $A$ and $B$ are both thermal states.
In this way, since no work can be extracted from thermal equilibrium [Eq.~\eqref{Eq:Aberg}], if there is any extractable work, it {\em must} come from the correlation.
Namely, we identify the extractable work from the correlation of a bipartite state $\rho_{AB}$ to be its extractable work {\em when local Hamiltonians make it locally thermal}.
When $\rho_{AB}$ is separable in the $AB$ bipartition and $[\rho_{AB},\rho_A\otimes\rho_B]=0$, one can use \AA berg's result Eq.~\eqref{Eq:Aberg} to show that the {\em one-shot \mbox{$\epsilon$-deterministic} extractable work from its correlation}, denoted by $W_{\rm corr,(1)}^\epsilon$, can be bounded as (again, $0<\epsilon\le1-1/\sqrt{2}$)
\begin{align}\label{Eq:WcorrState}
0\le \frac{W_{\rm corr,(1)}^\epsilon(\rho_{AB})}{k_BT\ln2} - D_0^\epsilon(\rho_{AB}\,\|\,\rho_A\otimes\rho_B)\le \log_2\frac{1}{1-\epsilon}.
\end{align}
Note that the quantity $W_{\rm corr,(1)}^\epsilon$ has temperature dependency, while we keep this implicit.

\subsection{Appendix III: Formal statement of the main result}
First, we introduce another thermodynamic task by imposing the following three additional constraints on the task defining $W^{\epsilon}_{\rm corr|\Theta,(1)}$ in Eq.~\eqref{Eq:ThermoTask}:
(i) Initial Hamiltonians are all fully degenerate.
(ii) Bipartite input $\eta_{AB}$ is always the maximally correlated classical state
$
\Phi_{AB} = \frac{1}{M}\sum_{m=0}^{M-1}\proj{m}_{A}\otimes\proj{m}_{B}
$. 
\mbox{(iii) $\norm{\Pi_M(\mathcal{N})\left({\id_A}/M\right) - {\id_A}/M}_1<2\epsilon$;} i.e., we only allow classical versions to generate informational non-equilibrium up to the order $O(\epsilon)$.
This task's genuinely transmitted energy is given by:
\begin{align}\label{Eq:AlternativeThermalTask}
&W_{\rm \Phi|\Theta,(1)}^\epsilon(\mathcal{N})\coloneqq\nonumber\\
&\quad\sup_{\substack{M\in\mathbb{N},\Pi_M\in\Theta_M\\\norm{\Pi_M(\mathcal{N})\left(\frac{{\id_A}}{M}\right) - \frac{{\id_A}}{M}}_1<2\epsilon}}W_{\rm corr,(1)}^\epsilon\left[(\Pi_M(\mathcal{N})_{A}\otimes\mathcal{I}_{B})(\Phi_{AB})\right].
\end{align}
Formally, Eq.~\eqref{Eq:MainResultInformalForm} can be stated as follows by using the two different tasks with figure-of-merits $W^{\epsilon}_{\rm corr|\Theta,(1)}$ and $W_{\rm \Phi|\Theta,(1)}^\epsilon$:
\begin{theorem}\label{Result:MainResult}
Consider a set of superchannels $\Theta$ and a background temperature $0~<~T~<~\infty$.
For a channel $\mathcal{N}$ and errors $0<\delta\le\omega<\epsilon\le1-1/\sqrt{2}$, we have that
\begin{align}
&W_{\rm corr|\Theta,(1)}^{\omega}(\mathcal{N})-k_BT\ln\frac{4\epsilon}{(\epsilon - \omega)^2(1-\omega)}\nonumber\\
&\quad\quad\quad\le(k_BT\ln2)C_{\Theta,(1)}^\epsilon(\mathcal{N})\le W_{\rm \Phi|\Theta,(1)}^{\epsilon+\delta}(\mathcal{N}).
\end{align}
\end{theorem}
We briefly sketch its proof here.
First, we obtain upper and lower bounds on $C_{\Theta,(1)}^\epsilon$ in terms of entropic quantities, which are summarised in companion paper Theorem~1~\cite{CompanionPRA2}.
After that, we relate these entropic bounds to $W_{\rm corr|\Theta,(1)}$ and $W_{\rm \Phi|\Theta,(1)}$ and complete the proof.
The full mathematical proof is given in companion paper Theorem~4~\cite{CompanionPRA2}, and here we focus on the implications of Theorem~\ref{Result:MainResult} to thermodynamics.

We comment that the upper bound is tight since all state-preparation channels achieve it.
In general, its attainability depends on the underlying set $\Theta$.
Also, since by definition $W_{\rm \Phi|\Theta,(1)}\le W_{\rm corr|\Theta,(1)}$, we immediately conclude that
{\em 
$W_{\rm \Phi|\Theta,(1)}\approx W_{\rm corr|\Theta,(1)} \approx (k_BT\ln2)C_{\Theta,(1)}$ up to one-shot error terms.
}
One can again apply the iid limit and obtain the asymptotic version:
{\em For a channel $\mathcal{N}$ and a background temperature \mbox{$0<T<\infty$}, we have that}
\begin{align}
(k_BT\ln2)C_\Theta(\mathcal{N}) = W_{\rm corr|\Theta}(\mathcal{N}) = W_{\Phi|\Theta}(\mathcal{N}),
\end{align}
where $W_{\Phi|\Theta}(\mathcal{N})\coloneqq\lim_{\epsilon\to0}\varliminf_{k\to\infty}\frac{1}{k}W_{\Phi|\Theta,(1)}^\epsilon\left(\mathcal{N}^{\otimes k}\right)$
can again be understood as the average transmitted energy in the asymptotic limit.
Finally, Theorem~\ref{Result:MainResult} provides a no-go result: {\em $W^{\epsilon}_{\rm corr|\Theta,(1)}$ and $W^{\epsilon}_{\rm \Phi|\Theta,(1)}$ are equivalent, up to one-shot error terms.}
Hence, the optimal transmitted energy $W^{\epsilon}_{\rm corr|\Theta,(1)}$ defined in Eq.~\eqref{Eq:ThermoTask} is actually achievable by using maximally correlated state $\Phi_{AB}$ as the bipartite input and classical versions $\Pi_M(\mathcal{N})$ that can hardly drive local systems out of thermal equilibrium---non-equilibrium effects cannot enhance the energy transmission, and no other classically correlated state can outperform $\Phi_{AB}$.
This also means that, in the present setting, arbitrary finite-energy initial Hamiltonians [as allowed in defining $W^\epsilon_{\rm corr|\Theta,(1)}$ in Eq.~\eqref{Eq:ThermoTask}] cannot outperform the fully degenerate ones [as required in defining $W^\epsilon_{\rm \Phi|\Theta,(1)}$ in Eq.~\eqref{Eq:AlternativeThermalTask}].

\bibliography{Ref.bib}

\end{document}